\documentclass[a4paper,UKenglish]{lipics}
 
\usepackage{microtype}
\usepackage{amsmath,amssymb}

\usepackage{gastex}
\usepackage{amsfonts}
\usepackage{amssymb}
\usepackage{amsmath}
\usepackage{latexsym}
\usepackage{amsthm}
\usepackage{eepic}
\usepackage{soul}

\def\ket#1{|#1\rangle}
\def\bra#1{\langle#1|}

\title{Universal Entanglers for Bosonic and Fermionic Systems\footnote{This work was partially supported by NSERC and CIFAR}}
\titlerunning{Universal Entanglers for Bosonic and Fermionic Systems} 

\author[1,2]{Joel Klassen}
\author[3,2,4]{Jianxin Chen}
\author[3,2]{Bei Zeng}
\affil[1]{Department of Physics, University of Guelph\\
  50 Stone Road East, Guelph, Ontario, Canada\\
  \texttt{joeldavidklassen@gmail.com}}
\affil[2]{Institute for Quantum Computing\\
  200 University Avenue West, Waterloo, Ontario, Canada}   
\affil[3]{Department of Mathematics \& Statistics, University of Guelph\\
  50 Stone Road East, Guelph, Ontario, Canada\\
  \texttt{\{jianxinc,zengb\}@uoguelph.ca}}
\affil[4]{UTS-AMSS Joint Research Laboratory for Quantum Computation and Quantum Information Processing\\
Academy of Mathematics and Systems Science, Chinese Academy of Sciences, Beijing, China}
\authorrunning{J. Klassen et. al.} 

\Copyright[by]{Jianxin Chen, Joel Klassen, and Bei Zeng}


\serieslogo{}
\volumeinfo
{Simone Severini (chair) \& Fernando Brandao (co-chair)}
  {2}
  {The 8th Conference on the Theory of Quantum Computation, Communication and Cryptography}
  {1}
  {1}
  {1}
\EventShortName{}
\DOI{10.4230/LIPIcs.xxx.yyy.p}

\begin{document}

\maketitle

\begin{abstract}
A universal entangler (UE) is a unitary operation which maps all pure product states to entangled states. It is known that
for a bipartite system of particles $1,2$ with a Hilbert space $\mathbb{C}^{d_1}\otimes\mathbb{C}^{d_2}$, a UE exists when $\min{(d_1,d_2)}\geq 3$ and $(d_1,d_2)\neq (3,3)$. It is also known that whenever a UE exists, almost all unitaries are UEs; however to verify whether a given unitary is a UE is very difficult since solving a quadratic system of equations is NP-hard in general. This work examines the existence and construction of UEs of bipartite bosonic/fermionic systems whose wave functions sit in the symmetric/antisymmetric subspace of $\mathbb{C}^{d}\otimes\mathbb{C}^{d}$. The development of a theory of UEs for these types of systems needs considerably different approaches from that used for UEs of distinguishable systems. This is because the general entanglement of identical particle systems cannot be discussed in the usual way due to the effect of (anti)-symmetrization which introduces ``pseudo entanglement'' that is inaccessible in practice. We show that, unlike the distinguishable particle case, UEs exist for bosonic/fermionic systems with Hilbert spaces which are symmetric (resp. antisymmetric) subspaces of $\mathbb{C}^{d}\otimes\mathbb{C}^{d}$ if and only if $d\geq 3$ (resp. $d\geq 8$). To prove this we employ algebraic geometry to reason about the different algebraic structures of the bosonic/fermionic systems. Additionally, due to the relatively simple coherent state form of unentangled bosonic states, we are able to give the explicit constructions of two bosonic UEs. Our investigation provides insight into the entanglement properties of systems of indisitinguishable particles, and in particular underscores the difference between the entanglement structures of bosonic, fermionic and distinguishable particle systems.
\end{abstract}

\section{Introduction}

Entanglement sits at the core of the counterintuitive and useful properties of quantum mechanics. At its inception Schr\"{o}dinger labeled entanglement ``the characteristic trait of quantum mechanics, the one that enforces its entire departure from classical lines of thought.'' \cite{Schrodinger} This observation remains true today, and with the advent of quantum computing, its practical consequences have never before been more real. However after decades of effort, entanglement remains poorly understood~\cite{nielsenchuang,AFOV08,HHH09}. A promising avenue for furthering our understanding of entanglement is cataloguing and analyzing the various means of generating it. There is a sense that those mechanisms which generate maximal amounts of entanglement, or most consistently generate entanglement, are especially enlightening because they serve as bounds on what can and can not be done, thus restricting our domain of inquiry. 

One outcome of this line of thought is the concept of a universal entangler (UE). A UE is a unitary operator which maps any non-entangled state to an entangled state~\cite{Chen:2008eq}. A UE can act as a useful tool, both theoretically and experimentally, due to its generality. This generality is derived from the fact that a UE admits any non-entangled quantum states. However this generality also makes demonstrating the properties of UEs very difficult. For instance, while it has been shown that UEs do exist for a system with Hilbert space $\mathbb{C}^{d_1}\otimes\mathbb{C}^{d_2}$ when $\min{(d_1,d_2)}\geq 3$ and $(d_1,d_2)\neq (3,3)$, proving this fact has been nontrivial, requiring techniques from algebraic geometry \cite{Chen:2008eq}. To date no elementary method is known which can achieve the same results. Additionally, although it has been shown that whenever UEs exist  almost all unitaries are UEs~\cite{Chen:2012}, explicit constructions of UEs remain ellusive. This is due to the fact that the problem of verifying whether a given unitary is a UE is in general intractable since the verification is equivalent to solving a quadratic system of equations which is hard in general \cite{FY80}. So far the only explicitly known UE is an example for the $(d_1,d_2)=(3,4)$, from an order $12$ Hadamard matrix~\cite{Chen:2012}. In general more advanced methods may be needed in order to construct UEs, as well as to verify their universality.

The theory of entanglement of systems of indistinguishable particles has garnered much attention during the past decade~\cite{SLM01,SCK+01,LZL+01,PY01,ESB+02,AFOV08}. The entanglement of systems of indistinguishable particles cannot necessarily be approached in the same way as the distinguishable particle case because the symmetry requirement of the wave functions (i.e. symmetrization for bosonic system and antisymmetrization for fermionic system) may introduce `pseudo entanglement' which is not accessible in practice \cite{ESB+02,SCK+01,SLM01,LZL+01,PY01}. It is now widely agreed that non-entangled states correspond to the coherent states $\ket{v}^{\otimes N}$~\cite{Puri:2001} for indistinguishable bosonic systems and to Slater determinants for indistinguishable fermionic systems~\cite{ESB+02,AFOV08}.
A natural line of inquiry is to identify the existence and construction of UEs for systems of indistinguishable particles. Indistinguishable bipartite bosonic/fermionic states are symmetric/antisymmetric states of the Hilbert space $\mathbb{C}^d\otimes\mathbb{C}^d$. This does not necessarily mean that the theory of UEs of distinguishable particles is readily generalizable to UEs for indistinguishable particles. Some obvious reasons for this are: 1. although almost all unitaries are UEs when $d>3$, the lack of understanding of explicit constructions prevents us from directly verifying whether there exist any UEs which are symmetric under particle permutation; 2. the definition of a non-entangled state for systems of indistinguishable fermions is dramatically different from that of systems of distinguishable particles (in fact, a single Slater determinant, when viewed as an antisymmetric distinguishable particle state, is indeed entangled).

This paper discusses the existence and construction of UEs for both indistinguishable bipartite bosonic (BUE) and fermionic (FUE) systems. Employing techniques in algebraic geometry, considering the different algebraic structures of the bosonic and fermionic systems, we show that, in contrast to the distinguishable particle case, BUEs exist for bosonic systems if and only if the single particle Hilbert space has dimension $d\geq 3$, and FUEs exist for fermionic systems if and only if the single particle Hilbert space has dimension $d\geq 8$. We also show, similarly to the distinguishable particle case, that for dimensions where BUEs/FUEs exist, almost all unitaries are BUEs/FUEs. Finally, because the unentangled states of indistinguishable bosonic systems are of a relatively simple coherent state form $\ket{v}\otimes\ket{v}$, which implies a hidden linear structure for the product states (i.e. the set of all single particle states $\ket{v}$ form a vector space), the construction of BUEs becomes significantly simpler. We have found a simple explicit construction of a BUE based on permutation matrices which holds for all $d \geq 3 $, and another one based on Householder-type gates \cite{horn1990matrix} which holds for all $d \geq 5$.  Unfortunately the explicit construction and verification of FUEs, like distinguishable particle UEs, remains a significantly more intractable problem. 


We believe that our investigation provides insight into the entanglement properties of identical particle systems, and in particular the different entanglement structures between bosonic, fermionic and distinguishable particle systems.

We organize our paper as follows. In section 2 we review some previously established results about UEs and provide some preliminaries about bosonic and fermionic systems to help establish our main results. In section 3 we give a proof for the existence and prevalence of BUEs, and give two explicit examples of their construction. In Section 4 we give a proof for the existence and prevalence of FUEs. Finally, in section 5, we provide a brief summary of our results and a discussion of future directions.


\section{Preliminaries}

This section provides preliminaries to help establish our main results for BUEs and FUEs. We first briefly review UEs for distinguishable particle systems established in~\cite{Chen:2008eq}. We then further briefly review basic entanglement theory for bosonic and fermionic systems.

\subsection{Universal entanglers}
For the case of distinguishable particles, it is known that any given quantum system is identified with some finite (or infinite) Hilbert space $\mathcal{H}$. Moreover, two unit vectors are indistinguishable if they differ only by a global phase factor. Hence, distinct pure states can be put in correspondence with ``rays'' in $\mathcal{H}$, or equivalently, points in the projective Hilbert space $\mathbb{P}(\mathcal{H})$. 

We consider pure states for bipartite systems, whose Hilbert space is $\mathbb{C}^{d_1}\otimes\mathbb{C}^{d_2}$. A bipartite
quantum state is a product state if $\ket{\psi}=\ket{\psi_1}\otimes\ket{\psi_2}$ for some $\ket{\psi_1}\in\mathbb{C}^{d_1}$
and $\ket{\psi_2}\in\mathbb{C}^{d_2}$. Otherwise, it is an entangled state. It is straightforward to see that the set of all the product states do not form a linear vector space, so one does not expect that the UE problem can be examined using basic tools from linear algebra. 

Instead, it is observed that the set of normalized product states in a composite system associated with $\mathbb{C}^{d_1}\otimes\mathbb{C}^{d_2}$ is isomorphic to a projective variety in $\mathbb{P}^{d_1d_2-1}$, a well studied object in algebraic geometry.  Before continuing, we need some basic notations and necessary background materials from algebraic geometry~\cite{Hartshorne:1983we}. 

For any positive integer $n$,  the set of all $n$-tuples from $\mathbb{C}$ is called an $n$-dimensional $\textit{affine space}$ over $\mathbb{C}$. An element of $\mathbb{C}^n$ is called a point, and if point $P=(a_1,a_2,\cdots,a_n)$ with $a_i\in \mathbb{C}$, then the $a_i$'s are called the coordinates of $P$.  Informally, an affine space is what is left of a vector space after forgetting its origin.

We define \textit{projective $n$-space}, denoted by $\mathbb{P}^n$, to be the set of equivalence classes of $(n+1)-$tuples $(a_0,\cdots,a_n)$ from $\mathbb{C}$, not all zero, under the equivalence relation given by $(a_0,\cdots,a_n)\sim(\lambda a_0,\cdots,\lambda a_n)$ for all $\lambda \in \mathbb{C}$, $\lambda\neq 0$. We use $[a_0:\cdots:a_n]$ to denote the projective coordinates of this point.

The \textit{polynomial ring} in $n$ variables, denoted by $\mathbb{C}[x_1,x_2,\cdots,x_n]$, is the set of polynomials in $n$ variables with coefficients in field $\mathbb{C}$.

A subset $Y$ of $\mathbb{C}^n$ is an \textit{algebraic set} if it is the common zeros of a finite set of polynomials $f_1,f_2,\cdots,f_r$ with $f_i\in \mathbb{C}[x_1,x_2,\cdots,x_n]$ for $1\leq i\leq r$, which is also denoted by $Z(f_1,f_2,\cdots,f_r)$.

One may observe that the union of a finite number of algebraic sets is an algebraic set, and the intersection of any family of algebraic sets is again an algebraic set. Therefore, by taking the open subsets to be the complements of algebraic sets, we can define a topology, called the \textit{Zariski topology} on $\mathbb{C}^n$.

A nonempty subset $Y$ of a topological space $X$ is called \textit{irreducible} if it cannot be expressed as the union of two proper closed subsets. The empty set is not considered to be irreducible.

An \textit{affine algebraic variety} is an irreducible closed subset of $\mathbb{C}^n$, with respect to the induced topology.

A notion of algebraic variety may also be introduced in projective spaces, called projective algebraic variety: a subset $Y$ of $\mathbb{P}^n$ is an \textit{algebraic set} if it is the common zeros of a finite set of homogeneous polynomials $f_1,f_2,\cdots,f_r$ with $f_i\in \mathbb{C}[x_0,x_1,\cdots,x_n]$ for $1\leq i\leq r$. We call open subsets of irreducible projective varieties quasi-projective varieties.

Observe that a product state in $\mathbb{C}^{d_1}\otimes \mathbb{C}^{d_2}$ can be written as the Kronecker product of a vector $v_1\in \mathbb{C}^{d_1}$ and another vector  $v_2\in \mathbb{C}^{d_2}$. Let's further write these vectors in the computational basis, say $v_1=\left(
\begin{array}{cccc}
x_1 ,& x_2, & \cdots &, x_{d_1} 
\end{array}
\right)$ and $v_2=\left(
\begin{array}{cccc}
y_1,  & y_2, & \cdots &, y_{d_2} 
\end{array}
\right)$. Their product state is a $d_1d_2$-dimensional vector
 
\begin{eqnarray*}
&\left(
\begin{array}{ccccccc}
z_1 , &  z_2, & \cdots  & z_{d_2}, & z_{d_2+1}, & \cdots &, z_{d_1d_2}   
\end{array}
\right)\\
=&\left(
\begin{array}{ccccccc}
x_1y_1, &  x_1y_2 , & \cdots&x_1y_{d_2},&x_2y_1,&\cdots&,x_{d_1}y_{d_2}  
\end{array}
\right)
\end{eqnarray*}

 Hence $z_{(i-1)d_2+j}=x_iy_j$ for any $1\leq i\leq d_1, 1\leq j\leq d_2$. It follows that 
 \begin{eqnarray*}
 z_{(i_1-1)d_2+j_1}z_{(i_2-1)d_2+j_2}&=&z_{(i_1-1)d_2+j_2}z_{(i_2-1)d_2+j_1}
 \end{eqnarray*}
 for any $1\leq i_1,i_2 \leq d_1, 1\leq j_1,j_2\leq d_2$. On the other hand, any $d_1d_2$-dimensional vector $(z_k)_{k=1}^{d_1d_2}$ satisfying the above polynomials can be written as the tensor product of $v_1\in \mathbb{C}^{d_1}$ and $v_2\in \mathbb{C}^{d_2}$~\cite{Hartshorne:1983we}. This implies that the set of normalized product states in $\mathbb{C}^{d_1}\otimes\mathbb{C}^{d_2}$ is isomorphic to a projective variety in $\mathbb{P}^{d_1d_2-1}$ which is called a ``Segre variety'' and denoted as $\Sigma_{d_1,d_2}$. This simple observation provides an algebraic geometric description of product states and entangled states. 

Therefore, a unitary operator $U$ acting on $\mathbb{C}^{d_1}\otimes \mathbb{C}^{d_2}$ is a UE if and only  
\begin{equation*}
U(\Sigma_{d_1,d_2})\bigcap \Sigma_{d_1,d_2}=\emptyset.
\end{equation*} 
From the geometric point of view, a UE will rotate the set of product states to another set which is completely void of product states.

In~\cite{Chen:2008eq}, it is proved that UEs exist if and only if $\min\{d_1,d_2\}\geq 3$ and $(d_1,d_2)\neq (3,3)$.  Surprisingly, it is further illustrated that a random unitary operator acting on such a bipartite system will even rotate the set of product states to another set which contains nothing but nearly maximally entangled states~\cite{Chen:2012}. 

Although it has been shown that a random unitary gate will almost surely be a UE of a bipartite quantum system  $\mathbb{C}^{d_1}\otimes \mathbb{C}^{d_2}$ if $\min\{d_1,d_2\}\geq 3$ and $(d_1,d_2)\neq (3,3)$, constructing an explicit UE for any bipartite quantum system is not that easy. One simple strategy is to randomly pick a unitary gate acting on $\mathbb{C}^{d_1}\otimes \mathbb{C}^{d_2}$, and then verify whether it is a UE by solving a family of polynomial equations. Unfortunately, there is no known efficient way to solve quadratic polynomial systems~\cite{FY80}. So far, explicit UEs are only known for $(d_1,d_2)=(3,4)$~\cite{Chen:2012}.

\subsection{Bosonic systems}

It is known that bosonic states lie in the $2$nd symmetric tensor power of $\mathbb{C}^d$, denoted by $\vee^2 \mathbb{C}^d$. A state in $\vee^2 \mathbb{C}^d$ is a product state if it can be written as some $\ket{\alpha}\otimes \ket{\alpha}$, i.e. it is a coherent state~\cite{PY01,ESB+02}. Any state which cannot be written as such a symmetric product form does demonstrate correlation which can be potentially used in quantum information processing~\cite{PY01}, and hence is considered entangled.

Any bipartite bosonic pure state is local unitarily equivalent to 
$\sum_{\alpha}\lambda_{\alpha}\ket{\alpha}\otimes\ket{\alpha}$~\cite{LZL+01,PY01}. 
This then indicates a hidden linear structure for bipartite bosonic pure states because the single particles states
$\ket{\alpha}$ form a vector space. 

From the algebraic geometric point of view, any bosonic product state $\ket{\alpha}\otimes \ket{\alpha}$ can be written as a vector with projective coordinates 
\begin{equation*}
[a_1a_1:a_1a_2:\cdots : a_1a_d: a_2a_1:a_2a_2:\cdots :a_2a_d: a_3a_1:\cdots :a_da_d]
\end{equation*}
where $[a_1:\cdots :a_d]$ are the projective coordinates of $\ket{\alpha}$.

Such points can be characterized by a family of polynomials again. In fact, the set of projective points with coordinates
\begin{equation*}
[a_1a_1:a_1a_2:\cdots : a_1a_d: a_2a_1:a_2a_2:\cdots :a_2a_d: a_3a_1:\cdots :a_da_d]
\end{equation*}
is obviously isomorphic to the set of the following points
\begin{equation*}
[a_1^2:a_2^2:\cdots : a_d^2: a_1a_2:a_1a_3:\cdots :a_1a_d: a_2a_3:\cdots :a_{d-1}a_d]
\end{equation*}
which is known as the Veronese variety in algebraic geometry \cite{Harris92}. 

Hence the set of bosonic product states corresponds to a special case of Veronese variety whose dimension is $d-1$. This fact will be used in our further investigation.

\subsection{Fermionic systems}

Consider the pure states of a bipartite fermionic system whose Hilbert space is the antisymmetric subspace of $\mathbb{C}^d\otimes\mathbb{C}^d$. The Pauli exclusion principle requires that $d\geq 2$. We denote  the $2$nd exterior 
power of $\mathbb{C}^d$, i.e. the antisymmetric subspace of $\mathbb{C}^d\otimes\mathbb{C}^d$ by $\wedge^2 \mathbb{C}^d$. For any $\ket{\alpha},\ket{\beta}\in \mathbb{C}^d$, we use the notation
\begin{equation}
\ket{\alpha}\wedge\ket{\beta}=\frac{1}{\sqrt{2}}(\ket{\alpha}\otimes\ket{\beta}-\ket{\beta}\otimes\ket{\alpha}),
\end{equation}
to denote a single Slater determinant.

A quantum state $\ket{\psi}$ in $\wedge^2 \mathbb{C}^{d}$ is said to be decomposable if it can be written as an exterior product of individual vectors from $\mathbb{C}^{d}$, i.e. there exists $\ket{\alpha},\ket{\beta}\in \mathbb{C}^d$ such that $\ket{\psi}=\ket{\alpha}\wedge\ket{\beta}$. Decomposable states are considered unentangled, as any correlation results purely
from the fermionic statistics, and so is not useful for quantum information processing~\cite{SLM01,SCK+01}. Any state which cannot be written in such a decomposable form does demonstrate correlation which can be potentially used in quantum information processing~\cite{SLM01,SCK+01} , and hence is considered to be entangled.

Any bipartite fermionic pure state is local unitarily equivalent to 
$\sum_{\alpha}\lambda_{i}\ket{\alpha_i}\wedge\ket{\beta_i}$~\cite{SLM01,SCK+01},
where $\ket{\alpha_i},\ket{\beta_i}\in \mathbb{C}^d$, $\langle{\alpha_i}|\beta_j\rangle=0$, $\langle\alpha_i|\alpha_j\rangle=\delta_{ij}$, and $\langle\beta_i|\beta_j\rangle=\delta_{ij}$.
This is an analogue of the Schmidt decomposition of a distinguishable particle 
system and hence is called the Slater decomposition. Similarly to the distinguishable particle 
case, the set of all decomposable states do not form a linear vector space, 
so one does not expect that the FUE problem can be examined using basic tools from linear algebra.

Again, let's look over the decomposable (or fermionic product) states from the algebraic geometric point of view. As we showed before, a decomposable state can be written as $\ket{\psi}=\ket{\alpha}\wedge\ket{\beta}$ where $\ket{\alpha}$ and $\ket{\beta}$ are two vectors in $\mathbb{C}^d$. Let $S_{\psi}$ be the $2$-dimensional subspace spanned by $\ket{\alpha}$ and $\ket{\beta}$. A different basis for $S_{\psi}$ will give a different exterior product, but the two exterior products will differ only by a nonzero scale. Ignoring the nonzero scale, any decomposable state corresponds to a $2$-dimensional subspace in $\mathbb{C}^d$ and vice versa. Hence the set of decomposable states is isomorphic to the set of $2$-dimensional subspaces which is known as a Grassmannian $G(2,d)$ \cite{Harris92}. It is not that obvious that $G(2,d)$ can be characterized by a set of polynomials, but it can be. The correspondence we have just shown is known as the Pl\"ucker embedding of a Grassmannian into a projective space:
\begin{equation*}
\tau: G(2,d) \rightarrow \mathbb{P}(\wedge^2 \mathbb{C}^d).
\end{equation*} 
This embedding satisfies certain simple quadratic polynomials and is called the Grassmann-Pl\"ucker relations (see e.g. p.
A III.172 Eq. (84-(J,H)) in \cite{Bourbaki:1970}, Prop 11-32 in \cite{Hasset:2007}, and \cite{ESB+02}). This implies the Grassmannian embeds as an algebraic variety of $\mathbb{P}(\wedge^2 \mathbb{C}^d)$.

\section{Bosonic Universal Entanglers}

\subsection{Existence and Prevalence}

Recall that a bosonic state in $\vee^2 \mathbb{C}^d$ is a product state if it can be written as $\ket{\alpha} \otimes \ket{\alpha}$ for some $\ket{\alpha}$. A quantum gate acting on $\vee^2\mathbb{C}^d$ is said to be a bosonic universal entangler (BUE) if it will map every product state to some entangled state.

Note that the set of product states of a bosonic system can also be characterized by a set of polynomials. Indeed, let $\Lambda=\{\ket{\alpha}\otimes \ket{\alpha}: \ket{\alpha}\in \mathbb{C}^d\}$, this is a precisely the Veronese variety \cite{Harris92}. Furthermore, $\Lambda$ is isomorphic to $\mathbb{C}^d$. For any $\ket{\psi}\in \vee^2\mathbb{C}^2$, let us denote $\text{rank} \ket{\psi}\equiv \min\{r: \ket{\psi}=\sum\limits_{i=1}^r \ket{a_i}\ket{a_i}\}$.

\begin{theorem}\label{theorem:bosonic}
There is a BUE acting on $\vee^2\mathbb{C}^d$ if and only if $d\geq 3$. Furthermore, when $d\geq 3$, almost every quantum gate acting on $\vee^2\mathbb{C}^d$ is a BUE.
\end{theorem}
\begin{proof}
For $d\leq 2$, we have
\begin{eqnarray}
\dim U(\Lambda)+ \dim \Lambda=2\dim \Lambda=2(d-1)\geq {d+1 \choose 2}-1=\dim \mathbb{P}(\vee^2 \mathbb{C}^d).
\end{eqnarray}

This implies there is no BUE for $\vee^2 \mathcal{C}^d$. This assertion follows from the dimension counting theorem which states that the intersection of any two projective varieties $\mathcal{A}$ and $\mathcal{B} \subseteq \mathbb{P}^m$ is nonempty if $\dim \mathcal{A}+\dim \mathcal{B}\geq m$. More specifically, we have $U(\Lambda)\bigcap \Lambda\neq \emptyset$.

On the other hand, consider the set of quantum gates acting on a system of two indistinguishable bosons. Any quantum gate acting on this system should be a symmetric gate, i.e., $S U S=U$, where $S$ is the swap operator.  Equivalently, $U$ is a quantum gate acting on $\vee^2 \mathbb{C}^d$.

Let $\mathcal{X}=\{\Phi|\Phi\in \mathcal{U}(\vee^2 \mathbb{C}^d),\Phi(\Lambda)\cap \Lambda\neq \emptyset\}$. Our aim is to show that $\mathcal{X}$ is a proper subset of $\mathcal{U}(\vee^2 \mathbb{C}^d)$. If this is so, then the existence of BUEs will be automatically guaranteed. 

Let's consider the Zariski topology on the projective space. In this setting, the unitary group $\mathcal{U}(\vee^2 \mathbb{C}^d)$ is Zariski dense in the general linear group $GL(\vee^2 \mathbb{C}^d)$~\cite{Schmitt:2008ww}. We further define $\mathcal{X}^{\prime}=\{\Phi|\Phi\in GL(\vee^2\mathbb{C}^d),\Phi(\Lambda)\cap \Lambda\neq \emptyset\}$. It is easy to see that $\mathcal{X}\subseteq \mathcal{X}^{\prime}$. 

The dimension of its  Zariski closure $\dim \overline{\mathcal{X}^{\prime}}$ is bounded by ${d+1\choose 2}^2-({d+1\choose 2}-1)+2(d-1)$. See Lemma~\ref{lemma:closure} in Appendix~\ref{appendix:closure} for details.

Now we prove the existence of a BUE as follows. If $U$ is not a BUE, $\mathcal{U}(\vee^2 \mathbb{C}^d)\subset \mathcal{X}^{\prime}$, then $GL(\vee^2\mathbb{C}^d)=\overline{\mathcal{U}(\vee^2\mathbb{C}^d)}\subset{\overline{\mathcal{X}^{\prime}}}$. However,  $\dim(\overline{\mathcal{X}^{\prime}})\leq {d+1\choose 2}^2-({d+1\choose 2}-1)+2(d-1) < {d+1\choose 2}^2 = \dim GL(\vee^2\mathbb{C}^d)$. This is a contradiction. So $\mathcal{U}(\vee^2\mathbb{C}^d)\not\subset \mathcal{X}^{\prime}$, i.e. a unitary operator $\Phi\in\mathcal{U}(\vee^2\mathbb{C}^d)$ with universal entangling power exists.

We will now show that $\mathcal{X}$ is not only a proper subset, but also a negligible subset of $\mathcal{U}(\vee^2\mathbb{C}^d)$.

$\mathcal{U}(\vee^2\mathbb{C}^d)$ is a locally compact Lie group of dimension ${d+1\choose 2}^2$. Recall that $\dim(\overline{\mathcal{X}^{\prime}})$ is at most ${d+1\choose 2}^2-({d+1\choose 2}-1)+2(d-1)<{d+1\choose 2}^2=\dim(\mathcal{U}(\vee^2\mathbb{C}^d))$. 

We have shown $\dim(\overline{\mathcal{X}^{\prime}}) <{d+1\choose 2}^2 = \dim(\mathcal{U}(\vee^2\mathbb{C}^d))$.  $\overline{\mathcal{X}^{\prime}}$ is Noetherian (i.e. any descending sequence of its closed subvarieties is stationary), then $\overline{\mathcal{X}^{\prime}}$ is a union of finitely many smooth subvarieties of $GL(\vee^2\mathbb{C}^d)$ with lower dimensions. Hence $\overline{\mathcal{X}^{\prime}}\cap \mathcal{U}(\vee^2\mathbb{C}^d)$ (which contains $\mathcal{X}^{\prime}\cap \mathcal{U}(\vee^2\mathbb{C}^d)$, the set
of our main interest) is a union of finite many submanifolds of $\mathcal{U}(\vee^2\mathbb{C}^d)$ with lower dimensions. Therefore, $\mathcal{X}^{\prime} \cap \mathcal{U}(\vee^2\mathbb{C}^d)$ is measure zero in $\mathcal{U}(\vee^2\mathbb{C}^d)$ which implies that a random unitary operator $U$ is almost surely a BUE.

\end{proof}


\subsection{Explicit Construction}

As we have shown in Theorem~\ref{theorem:bosonic}, a random unitary acting on $\vee^2\mathbb{C}^d$ will almost surely be a BUE. Hence we can pick an arbitrary unitary acting on $\vee^2\mathbb{C}^d$ and verify whether it will map some product state to another product state.  Recall that the set of product states in a bosonic system is isomorphic to $\mathbb{C}^d$. This will make it easier to verify whether a unitary is a BUE. Here we provide verifications of two different classes of BUEs.

\subsubsection{Householder-type Bosonic Universal Entanglers}

For $d\geq 5$ and any subspace $S\subset \vee^2\mathbb{C}^d$, let's consider the following gate $U=\mathbb{I}_{\vee^2\mathbb{C}^d}-2P_S$ where $P_S$ is a projection to some subspace $S$. These gates are known as Householder matrices in linear algebra  \cite{horn1990matrix} and they are widely used to perform QR decomposition.

A gate $U$ constructed in this way will be a BUE if the subspace $S$ is chosen properly to satisfy the following two constraints:
\begin{enumerate}
\item [1.] There is no product state in $S^{\perp}$.
\item [2.] $\text{rank} \ket{\psi}\geq 3$ for any $\ket{\psi}\in S$. 
\end{enumerate}

This claim can be proved by contradiction. Assume there are two product states $\ket{\psi}\ket{\psi}$ and $\ket{\phi}\ket{\phi}$ such that $(\mathbb{I}_{\vee^2\mathbb{C}^d}-2P_S)\ket{\psi}\ket{\psi}=\ket{\phi}\ket{\phi}$, we have $2P_S\ket{\psi}\ket{\psi}=\ket{\psi}\ket{\psi}-\ket{\phi}\ket{\phi}$. $P_S\ket{\psi}\ket{\psi}\neq 0$ since there is no product state in $S^{\perp}$. On the other hand, $P_S\ket{\psi}\ket{\psi}$ is a vector in $S$ which is a subspace completely void of states with rank no more than $2$. This contradicts our assumption.

In this subsection, we will construct a subspace $S$ to satisfy the above two constraints for any $d\geq 5$. A family of BUEs will follow immediately.

Let $S$ be the span of the following vectors.
\begin{eqnarray*}
&&\ket{11}+\ket{23}+\ket{32},\\
&&\ket{22}+\ket{34}+\ket{43},\\
&&\textrm{\ \ \ \ \ \ \ \ \ }\cdots,\\
&&\ket{d-2,d-2}+\ket{d-1,d}+\ket{d,d-1},\\
&&\ket{d-1,d-1}+\ket{d,1}+\ket{1,d},\\
&&\ket{d,d}+\ket{12}+\ket{21}.
\end{eqnarray*}

We first show there is no product state in $S^{\perp}$. Assume $\ket{\psi}\ket{\psi}\perp S$ where $\ket{\psi}=\sum\limits_{i=1}^d a_i \ket{i}$. The orthogonality implies the following equations.

\begin{eqnarray*}\label{eq:set1}
(E1)\left\{ \begin{array}{lll}
a_1^2+2a_2a_3&=&0,\\
a_2^2+2a_3a_4&=&0,\\
&\vdots&\\
a_d^2+2a_1a_2&=&0.
\end{array}\right.
\end{eqnarray*}

The only common solution to the above equations is  $(a_1,a_2,\cdots,a_d)=(0,0,\cdots,0)$ when $d\geq 3$. See Appendix~\ref{appendix:solution_1} for details.

Hence, there is no product state in $S^{\perp}$.

Next, we will verify that $\text{rank} \ket{\psi}\geq 3$ for any $\ket{\psi}\in S$.

Assume there is some state $\ket{\psi}\in S$ with rank no more than $2$. Let's say 
\begin{eqnarray}
&&c_1(\ket{11}+\ket{23}+\ket{32}),\\
&+&c_2(\ket{22}+\ket{34}+\ket{43}),\\
&+&\textrm{\ \ \ \ \ \ \ \ \ }\cdots ,\\
&+&c_d(\ket{d,d}+\ket{12}+\ket{21}),\\
&=&(x_1\ket{1}+x_2\ket{2}+\cdots +x_d\ket{d})(x_1\ket{1}+x_2\ket{2}+\cdots +x_d\ket{d}),\\
&+&(y_1\ket{1}+y_2\ket{2}+\cdots +y_d\ket{d})(y_1\ket{1}+y_2\ket{2}+\cdots +y_d\ket{d}).
\end{eqnarray}

Then we have the following equations.
\begin{eqnarray}\label{eq:set2}
(E2)\left\{\begin{array}{lll}
x_1^2+y_1^2&=&c_1,\\
x_2^2+y_2^2&=&c_2,\\
&\vdots&\\
x_d^2+y_d^2&=&c_d,\\
x_1x_2+y_1y_2&=&c_d,\\
x_2x_3+y_2y_3&=&c_1,\\
&\vdots&\\
x_{d}x_1+y_{d}y_1&=&c_{d-1},\\
x_ix_j+y_iy_j&=&0 \forall |i-j|\geq 2.
\end{array}\right.
\end{eqnarray}

There is no nonzero $(c_0,c_2,\cdots,c_d)$ satisfying the above equations when $d\geq 5$. See Appendix~\ref{appendix:solution_2} for details. Hence $\text{rank} \ket{\psi}\geq 3$ for any $\ket{\psi}\in S$. 

This implies that $U=I-2P_S$ is a bosonic universal entangler for any $d\geq 5$.

\subsubsection{Permutation Universal Entanglers}
Any product state can be written as the following.
\begin{eqnarray}
\ket{\phi}\ket{\phi}
&=&(\sum\limits_{i=1}^d a_i \ket{i})(\sum\limits_{j=1}^d a_j \ket{j})\\
&=&\sum\limits_{i,j=1}^d a_ia_j\ket{ij}\\
&=&\sum\limits_{i=1}^d a_i^2 \ket{ii}+\sum\limits_{1\leq i<j\leq d} \sqrt{2}a_ia_j(\frac{\ket{ij}+\ket{ji}}{\sqrt{2}}).
\end{eqnarray}

Any bosonic state $\ket{\psi}\in \vee^2 \mathbb{C}^d$ can be denoted as a ${d+1 \choose 2}$-dimensional vector 
\begin{equation*}
(x_{11},x_{22},\cdots,x_{dd},x_{12},\cdots,x_{1d}, x_{21},\cdots,x_{d-1 d})
\end{equation*}
since we can always write $\ket{\psi}$ as a linear combination of bosonic basis states
\begin{equation*}
x_{11}\ket{11}+x_{22}\ket{22}+\cdots+x_{dd}\ket{dd}+x_{12}\frac{\ket{12}+\ket{21}}{\sqrt{2}}+x_{13}\frac{\ket{13}+\ket{31}}{\sqrt{2}}+\cdots+x_{d-1,d}\frac{\ket{d-1,d}+\ket{d,d-1}}{\sqrt{2}}.
\end{equation*}

$\ket{\psi}$ is a product state if and only if there exists some nonzero vector $(a_1,a_2,\cdots, a_d)$ such that
\begin{eqnarray}
&&(x_{11},\cdots,x_{dd},x_{12},x_{13},\cdots,x_{1d}, x_{23},\cdots,x_{d-1 d})\\
&=&(a_1^2,\cdots,a_d^2, \sqrt{2} a_1a_2,\sqrt{2} a_1a_3,\cdots,\sqrt{2}a_1a_d,\sqrt{2}a_2 a_3,\cdots, \sqrt{2} a_{d-1} a_d).
\end{eqnarray}

A permutation matrix $U$ acting on the ${d+1\choose 2}$-dimensional vector space is certainly a bosonic quantum gate. 

For any $d\geq 3$, let's define a permutation matrix $U$ as the following:
\begin{eqnarray*}
U&=&\sum\limits_{i=1}^d(\frac{\ket{i,i+1}+\ket{i+1,i}}{\sqrt{2}})\bra{ii}+\sum\limits_{i=1}^d\ket{ii}(\frac{\bra{i,i+1}+\bra{i+1,i}}{\sqrt{2}})\\
&&+\sum\limits_{1\leq i<i+1<j\leq d}\frac{(\ket{ij}+\ket{ji})(\bra{ij}+\bra{ji})}{2}.
\end{eqnarray*}
Here the addition and subtraction are all modulo $d$, but the results range from $1$ to $d$.

$U$ is a unitary matrix since it is simply a rotation of the ${d+1\choose 2}$-dimensional vector space. Let's assume $U$ will map some (bosonic) product state to another (bosonic) product state. Without loss of generality, let's assume
\begin{eqnarray*}
U(\sum\limits_{i=1}^d a_i^2 \ket{ii}+\sum\limits_{1\leq i<j\leq d} \sqrt{2}a_ia_j(\frac{\ket{ij}+\ket{ji}}{\sqrt{2}}))=\sum\limits_{i=1}^d b_i^2 \ket{ii}+\sum\limits_{1\leq i<j\leq d} \sqrt{2}b_ib_j(\frac{\ket{ij}+\ket{ji}}{\sqrt{2}}).
\end{eqnarray*}

It follows that 
\begin{eqnarray*}
a_1^2&=&\sqrt{2} b_1b_2,\\
a_2^2&=&\sqrt{2} b_2b_3,\\
&\vdots& ,\\
a_d^2&=& \sqrt{2} b_db_1,\\
\sqrt{2} a_1a_2&=& b_1^2,\\
\sqrt{2} a_2a_3&=& b_2^2,\\
&\vdots& ,\\
\sqrt{2} a_da_1&=& b_d^2.
\end{eqnarray*}

Hence we have $\prod\limits_{i=1}^d a_i^2=(\sqrt{2})^d \prod\limits_{i=1}^d b_ib_{i+1}=\sqrt{2}^d \prod\limits_{i=1}^d b_i^2$. Similarly, $\prod\limits_{i=1}^d b_i^2=\sqrt{2}^d \prod\limits_{i=1}^d a_i^2$. The above two equations imply that there exists some $1\leq t\leq d$ such that $a_t=0$.

The equation $b_t^2=\sqrt{2} a_ta_{t+1}$ implies $b_t=0$. Then $a_{t-1}^2=\sqrt{2} b_{t-1}b_t=0$ will implies $a_{t-1}=0$. By repeating the above procedure, we will eventually have $a_i=0$ for any $1\leq i\leq d$. This contradicts our assumption that $U$ will map some (bosonic) product state to another (bosonic) product state. Hence $U$ is a bosonic universal entangler.

\section{Fermionic Universal Entanglers}

Given  a bipartite system of indisitinguishable fermions $\wedge^2 \mathbb{C}^{d}$, a $2$-vector in $\wedge^2 \mathbb{C}^{d}$ is said to be decomposable if it can be written as an exterior product of individual vectors from $\mathbb{C}^{d}$. Decomposable $2$-vectors are also considered to be unentangled states in this fermionic system.

We say a quantum gate $U$ is a fermionic universal entangler (FUE) if $U$ will transform every product state to some entangled state.

\begin{theorem}\label{theorem:fermionic}
There is some FUE acting on a bipartite system of indisitinguishable fermions $\wedge^2 \mathbb{C}^{d}$ if and only if  $d\geq 8$. Furthermore, almost every quantum gate acting on $\wedge^2 \mathbb{C}^{d}$is an FUE when $d\geq 8$.
\end{theorem}

\begin{proof}
Let $\Gamma_{d}=\{\ket{\phi}\in \wedge^2 \mathbb{C}^{d}: \ket{\phi}=\ket{\psi_1}\wedge \ket{\psi_2} \textrm{\ \ for some\ }\ket{\psi_1}, \ket{\psi_2}\in \mathbb{C}^{d} \}$. A quantum gate $U$ is an FUE if and only if 
\begin{eqnarray}
U(\Gamma_{d})\bigcap \Gamma_{d}= \emptyset.
\end{eqnarray}

Observe that decomposable $2$-vectors in $\wedge^2 \mathbb{C}^{d}$ correspond to weighted $2$-dimensional linear subspaces of $\mathbb{C}^{d}$. If we ignore the phase factor, decomposable $2$-vectors can be characterized by the Grassmannian of $2$-dimensional subspaces of $\mathbb{C}^{d}$, an algebraic subvariety of the projective space $\mathbb{P}(\wedge^2 \mathbb{C}^{d})$\cite{Hartshorne:1983we}. We will denote the Grassmannian of $r$-dimensional subspaces of $\mathbb{C}^d$ as $G(r,d)$.

First, we examine the necessary condition.

According to the intersection theorem, if $\dim U(\Gamma_{d})+\dim \Gamma_{d} \geq \dim \mathbb{P}(\wedge^2 \mathbb{C}^{d})$, or equivalently, $2\times2(d-2)=2\dim G(2,d)\geq {d \choose 2}-1$, then for any $U$, $U(\Gamma_{d})\bigcap \Gamma_{d}\neq \emptyset$. This inequality holds only for $2\leq d \leq 7$ which implies the fermionic universal entangling device does not exist for $d\leq 7$.

Now, let's look into the sufficient condition.

The set of quantum gates acting on a bipartite system of indisitinguishable fermions $\wedge^2 \mathbb{C}^d$ is the unitary group acting on $\wedge^2\mathbb{C}^d$, denoted as $\mathcal{U}(\wedge^2\mathbb{C}^d)$.

Similarly, let $\mathcal{Y}=\{\Phi|\Phi\in \mathcal{U}(\wedge^2 \mathbb{C}^d),\Phi(\Gamma_{d})\cap \Gamma_{d}\neq \emptyset\}$. We will show that $\mathcal{Y}$ is a proper subset in $\mathcal{U}(\wedge^2 \mathbb{C}^d)$. 

Again, let's consider the Zariski topology on the projective space. In this setting, the unitary group $\mathcal{U}(\wedge^2 \mathbb{C}^d)$ is Zariski dense in the general linear group $GL(\wedge^2 \mathbb{C}^d)$\cite{Schmitt:2008ww}. We further define $\mathcal{Y}^{\prime}=\{\Phi|\Phi\in GL(\wedge^2\mathbb{C}^d),\Phi(\Gamma_{d})\cap \Gamma_{d}\neq \emptyset\}$. It is easy to see $\mathcal{X}\subseteq \mathcal{Y}^{\prime}$. 

Similar to the proof of Lemma~\ref{lemma:closure} in Appendix~\ref{appendix:closure}, the dimension of $\mathcal{Y}^{\prime}$'s  Zariski closure $\dim \overline{\mathcal{Y}^{\prime}}$ is bounded by ${d\choose 2}^2-({d\choose 2}-1)+2\times2(d-2)$. 

Now we prove the existence of an FUE $U$ as follows. If it does not exist, $\mathcal{U}(\wedge^2 \mathbb{C}^d)\subset \mathcal{Y}^{\prime}$, then $GL(\wedge^2\mathbb{C}^d)=\overline{\mathcal{U}(\wedge^2\mathbb{C}^d)}\subset{\overline{\mathcal{Y}^{\prime}}}$. However,  $\dim(\overline{\mathcal{Y}^{\prime}})\leq {d\choose 2}^2-({d\choose 2}-1)+4(d-2) < {d\choose 2}^2 = \dim GL(\wedge^2\mathbb{C}^d)$. This is a contradiction. So $\mathcal{U}(\wedge^2\mathbb{C}^d)\not\subset \mathcal{Y}^{\prime}$, i.e. an FUE $\Phi\in\mathcal{U}(\wedge^2\mathbb{C}^d)$ exists.

Following the lines of the proof of Theorem~\ref{theorem:bosonic}, we can prove that $\mathcal{Y}$ is not only a proper subset, but also a neglectable subset in $\mathcal{U}(\wedge^2\mathbb{C}^d)$.
\end{proof}

\section{Summary and Discussion}
Employing properties of algebraic geometry, we have shown that for bipartite systems of indistinguishable bosons with a Hilbert space that is the symmetric subspace of $\mathbb{C}^d\otimes\mathbb{C}^d$, bosonic universal entanglers (BUEs) exist if and only if $d \geq 3$. Similarly, we have shown that for bipartite systems of indistinguishable fermions with a Hilbert space that is the antisymmetric subspace of $\mathbb{C}^d\otimes\mathbb{C}^d$, fermionic universal entanglers (FUEs) exist if and only if $d \geq 8$. These two results are in contrast to previous results regarding bipartite systems of distinguishable particles with a Hilbert space $\mathbb{C}^{d_1}\otimes\mathbb{C}^{d_2}$, for which universal entanglers exist if and only if min $(d_1, d_2) \geq 3$ and $(d_1,d_2) \neq (3,3)$.  This illustrates some of the important differences between the entanglement of systems of indistinguishable particles to the entanglement of systems of distinguishable particles. 

In contrast, we have illustrated one feature which holds for both systems of distinguishable and indistinguishable particles. Previous work has shown that, for systems of distinguishable particles, if a universal entangler exists for some Hilbert space, then almost all unitaries operating on that space are universal entanglers. We have shown that this result also holds for systems of indistinguishable bosons and fermions. However to verify whether or not a given bipartite unitary is a universal entangler is in general an intractable problem for both distinguishable particle systems and fermionic systems. This intractability arises from the fact that solving a system of quadratic equations is, in general, NP-hard.

Bosonic systems turn out to be special though. Because the set of all product states is isomorphic to a linear vector space, it is possible to use elementary methods to verify bosonic universal entanglers. We have given explicit constructions of two types of BUE, one is of the Householder type which is valid for $d\geq 5$ and the other is of a permutation type which is valid for $d\geq 3$. Both are very simple constructions.

It is our hope that our success in finding explicit constructions of BUEs will help inform the search for explicit constructions of both FUEs and UEs, problems which remain intractable in general. We can not rule out the possibiliy that there might be some other structure, beyond just the corresponding general algebraic varieties, which would provide some special family of explicitly verifiable UEs or FUEs. In fact, the explicit construction for the $(3,4)$ system from an order $12$ Hadamard matrix demonstrated in~\cite{Chen:2012} provides a hint of the possibility of such families. 

Another natural direction of inquiry is to explore the entangling power of these BUEs and FUEs. As demonstrated in~\cite{Chen:2012}, a random unitary is not only almost surely a UE, but it also almost surely maps the set of product states to another set which contains nothing but nearly maximally entangled states, with respect to almost any kind of entanglement measure. One would expect similar properties for BUEs and FUEs. However to go further in that direction one would need to first establish reasonable entanglement measures for bosonic and fermionic systems (see, e.g. entanglement measures discussed in~\cite{ESB+02}).

Finally, it would be useful to generalize these results to multipartite bosonic and fermionc systems. Our guess is that the bosonic systems might remain easy to solve since they retain the nice property that the set of all product states is isomorphic to a linear vector space. The fermionic case is expected to be much more complicated given that in the multipartite case even the Grassmann-Pl\"ucker relations themselves are harder to describe~\cite{Bourbaki:1970,Hasset:2007,ESB+02}. We would leave these cases for future investigation.

\subparagraph*{Acknowledgements}
JK is supported by NSERC. JC is supported by NSERC, UTS-AMSS Joint Research Laboratory for Quantum Computation and Quantum Information Processing and NSF of China (Grant No. 61179030). BZ is supported by NSERC and CIFAR.

\appendix
\section{There Is No Nonzero Solution For Polynomial System~(E1)}\label{appendix:solution_1}
Here we will show there is no nonzero solution $(a_1,\cdots, a_d)$ satisfying the following equations.
\begin{eqnarray*}\label{eq:set1}
(E1)\left\{ \begin{array}{lll}
a_1^2+2a_2a_3&=&0,\\
a_2^2+2a_3a_4&=&0,\\
&\vdots&\\
a_d^2+2a_1a_2&=&0.
\end{array}\right.
\end{eqnarray*}

Assume $a_i\neq 0$, then $a_{i+1}, a_{i+2}$ are nonzero. This follows all $a_i$'s are nonzero.
\begin{eqnarray}
\mathop{\Pi}\limits_{i=1}^d a_i^2=\mathop{\Pi}\limits_{i=1}^d(-2 a_{i+1}a_{i+2})=(-2)^d\mathop{\Pi}\limits_{i=1}^d a_i^2.
\end{eqnarray}

This implies $\mathop{\Pi}\limits_{i=1}^d a_i^2=0$. Hence it is a contradiction. 

Therefore, the only solution to this polynomial system is $(a_1,\cdots,a_d)=(0,\cdots,0)$.

\section{There Is No Nonzero Solution For Polynomial System~(E2)}\label{appendix:solution_2}
Here we will show there is no nonzero solution $(c_0,c_1,\cdots,c_d)$ satisfying the following equations.

\begin{eqnarray}\label{eq:set2}
(E2)\left\{\begin{array}{lll}
x_1^2+y_1^2&=&c_1,\\
x_2^2+y_2^2&=&c_2,\\
&\vdots&\\
x_d^2+y_d^2&=&c_d,\\
x_1x_2+y_1y_2&=&c_d,\\
x_2x_3+y_2y_3&=&c_1,\\
&\vdots&\\
x_{d}x_1+y_{d}y_1&=&c_{d-1},\\
x_ix_j+y_iy_j&=&0 \forall |i-j|\geq 2.
\end{array}\right.
\end{eqnarray}

For $d\geq 5$, let's assume there is some $1\leq i\leq d$ such that $x_i=0$ and $y_i\neq 0$. 

It follows from $x_i x_j+y_iy_j=0$ for any $|j-i|\geq 2$ that $y_j=0$ for any $|j-i|\geq 2$. So, $y_{i+2}=y_{i+3}=0$.

Then, $0\neq x_i^2+y_i^2=x_{i+1}x_{i+2}+y_{i+1}y_{i+2}=x_{i+1}x_{i+2}$. This implies $x_{i+1},x_{i+2}\neq 0$.

From $x_{i+1}x_{i+3}+y_{i+1}y_{i+3}=0$, we have $x_{i+3}=0$. So, $0\neq x_{i+2}^2+y_{i+2}^2=x_{i+3}x_{i+4}+y_{i+3}y_{i+4}=0$. This is a contradiction.

So, for any $1\leq i\leq d$, we have $x_i=y_i=0$ or $x_i y_i\neq 0$.

Let's look into the various situations.

\begin{enumerate}
\item [1.]  There is some $i$ such that $x_i=y_i=0$. Then we have $x_{i-1}^2+y_{i-1}^2=0$. If $x_{i-1}=y_{i-1}=0$, we consider $x_{i-2}^2+y_{i-2}^2=0$. By repeating this procedure, if all $x_j$'s,$y_j$'s are not all zero, we will find some $i'$ such that $y_{i'}=ix_{i'}$ or $y_{i'}=-i x_{i'}$ and $x_{i'+1}=y_{i'+1}=0$. We will further have $y_{i'+k}=\pm x_{i'+k}$ for any $k=2,\cdots, d-1$. This implies that $(c_1,\cdots,c_d)=0$. 
\item [2.] All $x_i$'s, $y_i$'s are nonzero. For any fixed $i$, $\frac{y_j}{x_j}=-\frac{x_i}{y_i}$ for any $j=i+2,\cdots,i+d-2$.  This implies $\frac{y_k}{x_k}$ is a constant $i$ or $-i$. This also implies $(c_1,\cdots, c_d)=0$.
\end{enumerate} 

\section{Proof of Lemma~\ref{lemma:closure}}\label{appendix:closure}
\begin{lemma}\label{lemma:closure}
$\dim(\overline{\mathcal{X}^{\prime}})\leq {d+1\choose 2}^2-({d+1\choose 2}-1)+2(d-1)$, where $\overline{\mathcal{X}^{\prime}}$ is the Zariski closure of $\mathcal{X}^{\prime}$.
\end{lemma}

The following technical lemmas will be needed.

\begin{lemma}[\cite{Tauvel:2005us}]\label{lemma:dim1}
If $Z_1$ and $Z_2$ are both irreducible varieties over $\mathbb{C}$, and $\phi:Z_1\rightarrow Z_2$ is a dominant morphism, then $\dim(Z_2)\leq \dim(Z_1)$. Here, dominant means $\Phi(Z_1)$ is dense in $Z_2$.
\end{lemma}

\begin{lemma}[\cite{Tauvel:2005us}]\label{lemma:dim2}
If $Z_1$ and $Z_2$ are both varieties over $\mathbb{C}$, and $\phi:Z_1\rightarrow Z_2$ is a morphism, then $\dim(Z_1)\leq \dim(Z_2)+\max\limits_{z\in Z_2}{\dim(\phi^{-1}(z))}$.
\end{lemma}

Lemma~$\ref{lemma:dim1}$ and Lemma~\ref{lemma:dim2} establish a connection between the dimensions of domain and codomain of a variety morphism.

\begin{proof}
We have a morphism $F:GL(\vee^2\mathbb{C}^d)\times \mathbb{P}^{{d+1\choose 2}-1} \rightarrow \mathbb{P}^{{d+1\choose 2}-1}$ which is just the left
action of $GL(\vee^2\mathbb{C}^d)$ on  $\mathbb{P}^{{d+1\choose 2}-1}$, defined by $F(g,[w])=[g\cdot w]$.

We let $y_0=(1,0,\cdots,0)$ be a row vector with $d+1\choose 2$ entries, and for any given $y_1$, $y_2\in \mathbb{P}^{{d+1\choose 2}-1}$, we choose proper $g_1$ and $g_2\in GL(\vee^2\mathbb{C}^d)$, such that $[g_1\cdot y_0]=[y_1]$ and $[g_2\cdot y_0]=[y_2]$. Then we have 
\begin{eqnarray}
[g\cdot y_2]=[y_1] \iff  [g g_2\cdot y_0]=[g g_1\cdot y_0] \iff & [g_1^{-1}gg_2\cdot y_0]=[y_0].
\end{eqnarray}

From the above observations, F has the following property: for any $y_1$, $y_2\in \mathbb{P}^{{d+1\choose 2}-1}$,
$F^{-1}(y_2)\cap \{GL(\vee^2\mathbb{C}^d)\times \{y_1\}\}\cong \{\left(\begin{array}{cc}
z_1& \alpha\\
0 & g'
\end{array}\right):z_1\in \mathbb{C}\backslash \{0\}, g'\in GL({d+1\choose 2}-1), \alpha\in \mathbb{C}^{{d+1\choose 2}-1} \ is \ a \ row \
vector.\}$. Hence $\dim(F^{-1}(y_2)\cap {GL(\vee^2\mathbb{C}^d)\times
\{y_1\}})={d+1\choose 2}^2-({d+1\choose 2}-1)$.

Let $P_1$, $P_2$ be projections of $GL(\vee^2\mathbb{C}^d)\times \mathbb{P}^{{d+1\choose 2}-1}$ to $GL(\vee^2\mathbb{C}^d)$, $\mathbb{P}^{{d+1\choose 2}-1}$ respectively.
Now we only look at $GL(\vee^2\mathbb{C}^d)\times \Lambda\subseteq GL(\vee^2\mathbb{C}^d)\times \mathbb{P}^{{d+1\choose 2}-1}$, to get
$F:GL(\vee^2\mathbb{C}^d)\times \Lambda\rightarrow \mathbb{P}^{{d+1\choose 2}-1}$. Then we have a characterization of $\mathcal{X}^{\prime}$:
$\mathcal{X}^{\prime}=P_1F^{-1}(\Lambda)$. In fact
\begin{align*}
& g\in \mathcal{X}^{\prime} \\
\iff & g(\Lambda)\cap \Lambda\neq \emptyset  \\
\iff &\exists z_1, z_2\in \Lambda, s.t. g(z_1)=z_2 \\
\iff &\exists z_1, z_2\in \Lambda, s.t. (g, z_1)\in F^{-1}(z_2)\\
\iff &\exists z_2\in \Lambda, s.t. g\in P_1F^{-1}(z_2)\\
\iff & g\in P_1 F^{-1}(\Lambda).
\end{align*}
So $\overline{\mathcal{X}^{\prime}}\subseteq GL(\vee^2\mathbb{C}^d)$ is the Zariski closure of $\mathcal{X}^{\prime}$, which is also an algebraic variety.

Next, we assert that 
$P_1: F^{-1}(\Lambda)\rightarrow \overline{\mathcal{X}^{\prime}}$ is a dominant morphism.

Furthermore, consider $\Psi: F^{-1}(\Lambda)\rightarrow \Lambda\times \Lambda$ given by $\Psi(g,[z])=([z],[g\cdot z])$.

For $\forall z_1\in \Lambda$, $z_2\in \Lambda$, we have $\Psi^{-1}(z_1,z_2)= (g_2 T g_1^{-1}, z_1)$, where $T=\{\left(\begin{array}{cc}
z_0& \alpha\\
0 & g'
\end{array}\right):z_0\in \mathbb{C}\backslash \{0\}, g'\in GL({d+1\choose 2}-1), \alpha\in \mathbb{C}^{{d+1\choose 2}-1} \ is \ a \ row \
vector\}$, and $g_1, g_2\in GL(\vee^2\mathbb{C}^d)$, s.t. $g_1(y_0)=z_1$, $g_2(y_0)=z_2$.
So this is a dominant morphism. Then we obtain
\begin{align*}
\dim(F^{-1}(\Lambda))
\leq & \dim(T)+\dim(\Lambda\times \Lambda)\\
=& {d+1\choose 2}^2-({d+1\choose 2}-1)+\dim(\Lambda)+\dim(\Lambda).
\end{align*}

It is required in Lemma~\ref{lemma:dim1} that varieties $Z_1$ and $Z_2$ be irreducible. Actually, this condition can be weakened. Lemma~\ref{lemma:dim1} is still true for the more general case that $Z_1$ and $Z_2$ are closed subsets of irreducible varieties\cite{Hartshorne:1983we}. Through this approach, we can fill out the gap and apply this lemma without danger of confusion. Indeed, the irreduciblity of $Z_1$ and $Z_2$ really holds, but verification of this is not easy. 

Then from Lemma~\ref{lemma:dim1} and Lemma~\ref{lemma:dim2}, we will have
\begin{align*}
\dim(\overline{\mathcal{X}^{\prime}})\leq &\dim(F^{-1}(\Lambda))\\
\leq &
{d+1\choose 2}^2-({d+1\choose 2}-1)+\dim(\Lambda)+\dim(\Lambda)\\
=& {d+1\choose 2}^2-({d+1\choose 2}-1)+2(d-1).
\end{align*}
\end{proof}

\bibliography{lipics-UE}







\end{document}